\documentclass[12pt]{article}

\usepackage[T1]{fontenc}
\usepackage{kpfonts}

\usepackage{setspace}
\usepackage[margin=1.25in]{geometry}

\usepackage[dvipsnames]{xcolor}
\usepackage{pdfpages}
\usepackage{float}
\usepackage{multirow}
\usepackage{caption}
\usepackage{subcaption}
\usepackage{epigraph}

\usepackage{mathtools}
\usepackage{amsthm}
\usepackage{aliascnt} % lets cleveref distinguish theorem-like environments
\usepackage{accents}
\usepackage{dutchcal}
\usepackage{bbm}

\usepackage{enumitem}
\usepackage{sgame}

\usepackage{tikz}
\usepackage{tikz-cd}
\usetikzlibrary{calc,shapes,arrows}

\usepackage{tcolorbox}
\usepackage{listings}

\usepackage[normalem]{ulem}

\usepackage[round]{natbib}
\DeclareMathOperator{\ext}{ext}

\providecommand{\keywords}[1]{\textbf{\textit{Keywords:}} #1}
\providecommand{\jel}[1]{\textbf{\textit{JEL Classifications:}} #1}

\newtheorem{theorem}{Theorem}[section]

\newaliascnt{proposition}{theorem}
\newtheorem{proposition}[proposition]{Proposition}
\aliascntresetthe{proposition}

\newaliascnt{lemma}{theorem}
\newtheorem{lemma}[lemma]{Lemma}
\aliascntresetthe{lemma}

\newaliascnt{corollary}{theorem}
\newtheorem{corollary}[corollary]{Corollary}
\aliascntresetthe{corollary}

\newaliascnt{claim}{theorem}
\newtheorem{claim}[claim]{Claim}
\aliascntresetthe{claim}

\theoremstyle{definition}

\newaliascnt{definition}{theorem}
\newtheorem{definition}[definition]{Definition}
\aliascntresetthe{definition}

\newaliascnt{example}{theorem}

\aliascntresetthe{example}

\newaliascnt{condition}{theorem}

\aliascntresetthe{condition}

\newtheorem{assumption}{Assumption}

\newaliascnt{remark}{theorem}
\newtheorem{remark}[remark]{Remark}
\aliascntresetthe{remark}

\newaliascnt{remarks}{theorem}

\aliascntresetthe{remarks}

\newaliascnt{aside}{theorem}

\aliascntresetthe{aside}

\newaliascnt{note}{theorem}

\aliascntresetthe{note}

\usepackage{hyperref}

\hypersetup{
    colorlinks=true,
    linkcolor=OrangeRed,
    filecolor=Thistle,
    urlcolor=Thistle,
    citecolor=Thistle,
}

\usepackage[nameinlink]{cleveref}

\crefname{theorem}{theorem}{theorems}
\Crefname{theorem}{Theorem}{Theorems}

\crefname{proposition}{proposition}{propositions}
\Crefname{proposition}{Proposition}{Propositions}

\crefname{lemma}{lemma}{lemmas}
\Crefname{lemma}{Lemma}{Lemmas}

\crefname{corollary}{corollary}{corollaries}
\Crefname{corollary}{Corollary}{Corollaries}

\crefname{claim}{claim}{claims}
\Crefname{claim}{Claim}{Claims}

\crefname{definition}{definition}{definitions}
\Crefname{definition}{Definition}{Definitions}

\crefname{example}{example}{examples}
\Crefname{example}{Example}{Examples}

\crefname{assumption}{assumption}{assumptions}
\Crefname{assumption}{Assumption}{Assumptions}

\crefname{condition}{condition}{conditions}
\Crefname{condition}{Condition}{Conditions}

\crefname{question}{question}{questions}
\Crefname{question}{Question}{Questions}

\crefname{remark}{remark}{remarks}
\Crefname{remark}{Remark}{Remarks}

\crefname{remarks}{remarks}{remarks}
\Crefname{remarks}{Remarks}{Remarks}

\crefname{aside}{aside}{asides}
\Crefname{aside}{Aside}{Asides}

\crefname{note}{note}{notes}
\Crefname{note}{Note}{Notes}

\newcommand{\secref}[1]{\hyperref[#1]{\S\ref*{#1}}}

\definecolor{backcolour}{rgb}{0.63, 0.79, 0.95}

\lstdefinestyle{mystyle}{
  backgroundcolor=\color{backcolour},
  basicstyle=\ttfamily\footnotesize,
  breakatwhitespace=false,
  breaklines=true,
  captionpos=b,
  keepspaces=true,
  numbers=left,
  numbersep=5pt,
  showspaces=false,
  showstringspaces=false,
  showtabs=false,
  tabsize=2
}

\begin{document}
\author{Konstantin von Beringe \and Mark Whitmeyer \thanks{Arizona State University; \href{kvonberi@asu.edu}{kvonberi@asu.edu} and \href{mailto:mark.whitmeyer@gmail.com}{mark.whitmeyer@gmail.com}.}}
\title{Robust Welfare under Imperfect Competition}
\maketitle

\begin{abstract}
We study welfare analysis for policy changes when supply and demand behavior are only partially known. We augment the robust approach pioneered by \citet{KangVasserman2025} by incorporating the supply side. We posit intervals of feasible pass-through and conduct (market-power) parameters, then apply them to two equilibrium snapshots to characterize the extremal supply-side terms entering welfare. We show that the supply-side bounds are attained by inverse pass-through functions that take only the two endpoint values of the specified interval, separated by a single price cutoff. Combining these supply-side extrema with demand-side shape restrictions, we produce simple bounds for consumer surplus, producer surplus, total surplus, and deadweight loss.
\end{abstract}
\\
\jel{C14; C51; D04; D61; F14; H21; Q48}

\section{Introduction}

Empirical evaluations of tax and regulatory policies often rest on functional-form assumptions that impute unobserved objects--most commonly, the demand curve connecting two observed price-quantity points. Recent work by \citet{KangVasserman2025} (\textbf{KV}) formalizes how to assess the robustness of welfare conclusions to such demand-side assumptions by asking how different the true curve must be from a maintained specification to overturn a qualitative conclusion (e.g., that a policy is net beneficial). Their framework shows that, holding fixed the observed endpoints and shape class, welfare is maximized or minimized by simple, canonical demand curves.

This paper builds on KV's work: we develop a \emph{supply-side} analog of robust welfare analysis for discrete (non-marginal) policy changes. Our starting point is the same minimal evidence set--two snapshots around a tax: \(\left(p_0,Q_0,\tau_0\right)\) and \(\left(p_1,Q_1,\tau_1\right)\)--but we add two weak restrictions on supply behavior without imposing a parametric cost function: (i) \emph{an interval of feasible pass-through rates} through which the unobserved price schedule \(p(\tau)\) can move, and (ii) \emph{an interval of feasible conduct (market-power) parameters} that link price-cost markups to the familiar marginal-revenue correction via the conduct identity introduced in \cite{bresnahan1982oligopoly} (see also \cite{genesovemullin1998} and \cite{weylfabinger2013}):
\[
p-\tau-\mathrm{MC}(Q)=\kappa(Q)\left[-P'(Q)Q\right],
\]
where \(\kappa(Q)\) is a conduct parameter nesting, e.g., perfect competition, monopoly, and (symmetric) Cournot competition.

These ingredients let us bound consumer surplus (\(\mathrm{CS}\)), producer surplus (\(\mathrm{PS}\)), total surplus (\(\mathrm{TS}\)), and deadweight loss (\(\mathrm{DWL}\)), with tax revenue (\(\mathrm{TR}\)) pinned down directly by the two observed snapshots. We remain agnostic about functional forms beyond the demand-shape classes developed by \citet{KangVasserman2025} and impose only very weak regularity on the unobserved price schedule.

The key technical observation is an accounting identity that isolates all dependence on the unobserved pass-through schedule in a single linear integral term:
\[
\int_{p_0}^{p_1}Q(p)u(p)dp,\quad \text{where}\quad u(p)\in[\alpha,\beta]\quad\text{and}\quad \int_{p_0}^{p_1}u(p)dp=\Delta\tau,
\]
where \(u\) is the inverse pass-through, i.e., the change-of-variables weight induced by the unobserved price schedule via \(d\tau=u(p)dp\), and \(\alpha\) and \(\beta\) are the inverses of the exogenously given pass-through bounds.

Because welfare objects are linear combinations of \(\int Qdp\), the schedule-weighted term \(\int Qudp\), and a linear conduct term, the overall problem is bilinear in \((Q,u)\). We show that the only relevant supply-side choice is how to place the inverse pass-through weight \(u\) across prices, subject to its bounds and its integral constraint. Since \(Q\) is decreasing, a simple (Hardy-Littlewood) rearrangement argument says that to maximize (minimize) \(\int Qudp\), we should place the high value of \(u\) where demand is largest (smallest). Thus, the extremal \(u\)s take only the endpoint values \(\alpha\) and \(\beta\), with a single cutoff in price. On the demand side, the selected inverse pass-through schedule and welfare object induce a weighted integral of transformed demand. If this induced weight has one sign, KV's extremal demand curves apply directly. 

If the weight changes sign once, we cannot appeal to KV; instead, we show that an extremum is attained by a transformed demand curve that EATS, standing for ``takes Endpoints with At most Two Switches.'' Informally, after writing demand in the units used for the chosen shape restriction--for example, levels, logs, or elasticity units--the slope of the transformed demand curve takes only the lower and upper values allowed by the restriction, and switches between them at most twice. Thus, even when the welfare weight changes sign, the demand-side problem reduces to a simple class of candidate curves. \secref{sec:demand} makes this definition formal.

Combining these steps delivers a sharp and tractable characterization. For each welfare object and conduct endpoint, a researcher needs only to evaluate the two extremal inverse pass-through schedules. Conditional on each schedule, the demand-side extremum is given by either the KV one or one that EATS, obtained from a single-dimensional optimization over the value of transformed demand at the sign-change point. Thus, we have a recipe that produces tight bounds for consumer surplus, producer surplus, total surplus, and deadweight loss (and any convex combination thereof), given a pick of demand-shape class and pass-through/conduct intervals.

Our treatment also clarifies the role of market power. The conduct identity can be rewritten as
\[
p-\tau-\mathrm{MC}(Q)=\kappa(Q)\left[-P'(Q)Q\right]=L(p)p,\qquad L(p)=\frac{\kappa(Q)}{\left|\eta(p)\right|},
\]
for the price-elasticity of demand \(\eta\), so conduct \(\kappa\) and inverse elasticity jointly govern the Lerner index \(L(p)\) over the price support. This yields transparent upper and lower producer surplus and total surplus bounds: lower welfare bounds use the upper conduct bound and the inverse pass-through schedule that maximizes the \(Q\)-weighted exposure to high \(u\), whereas upper welfare bounds use the lower conduct bound and the opposite inverse pass-through schedule. When one has external moments on elasticities or markups, they can be plugged directly into our formulas.

\paragraph{Roadmap.} We finish this section by discussing related work. Next, in \secref{sec:themodel} we introduce the model and go through preliminary results in \secref{sec:prelim}. In \secref{sec:supply} we tackle the supply-side problem; in \secref{sec:demand}, the demand-side; and we combine these analyses to get our main result in \secref{sec:mainresult}. All proofs omitted from the main text may be found in the Appendix (\ref{appendix}).

\subsection{Related Work}\label{sec:relatedwork}

Our paper contributes to a growing literature that studies welfare analysis under minimal functional-form assumptions. The closest antecedent is \citet{KangVasserman2025}, who develop \emph{robust} bounds for finite policy changes using shape restrictions on demand and show that familiar functional forms (linear, CES, exponential etc.) are extremal within various canonical families. We build on this program by bringing the supply side into the bounding exercise: we integrate conduct information (via Lerner and markup relationships) and bounds on pass-through into welfare formulas. In doing so, we deliver sharp bounds that combine (i) demand-shape extrema with (ii) extremal inverse pass-through functions consistent with our pre-specified pass-through intervals. A producer-surplus identity makes the roles of conduct and inverse pass-through transparent.

A second strand is the tax-incidence and pass-through literature under imperfect competition. \citet{weylfabinger2013} synthesize incidence principles beyond perfect competition and link global incidence to the integral of local pass-through (their various ``principles of incidence 5''). Our approach operationalizes this integral for policy changes by treating the (unobserved) tax schedule as a weighted integral over demand, with weights pinned down by bounds on pass-through. This makes explicit how pass-through heterogeneity over the price interval maps into welfare ranges. %We also connect to work comparing specific and \emph{ad valorem} interventions under market power. Classic analyses show that the two instruments can yield different pass-through and incidence patterns under imperfect competition \citep{DelipallaKeen1992,andersondepalmakreider2001}.

On the supply side, our conduct intervals relate to the new empirical IO tradition that infers market power and tests conduct. Foundational contributions include the Lerner index as a measure of monopoly power \citep{lerner1934} and the identification and testing of oligopoly conduct in both homogeneous and differentiated-product settings \citep{bresnahan1982oligopoly,nevo1998,genesovemullin1998}. We exploit the standard connection \(L=\kappa/|\eta|\) to translate bounds on conduct (or Lerner indices) and on elasticities into transparent wedges in our welfare accounting.

More broadly, our robustness perspective complements two adjacent literatures. First, the \textbf{sufficient-statistics} approach emphasizes using a few elasticities or reduced-form objects to conduct welfare analysis under maintained structure \citep{chetty2009}. Our results can be read as delivering ranges when only inequality-type information about these objects is credible. Second, the \textbf{partial identification} program (e.g., \citealp{manski2003}) advocates reporting identified sets rather than point estimates under weak assumptions. Our welfare bounds are a structural incarnation of this principle tailored to price-quantity data and policy shocks. Relatedly, robust pricing and mechanism-design work studies optimal behavior under model uncertainty, e.g., robust monopoly pricing under worst-case payoffs or regret \citep{bergemannschlag2011,condorelli2022lower,le2024robust}. We shift the focus from the seller's objective to social welfare under constrained knowledge of demand and supply. There are also other papers bounding objects of interest under monopoly and imperfect competition \citep{condorelli2022surplus,condorelli2022lower,le2024robust}.

On a technical level, our work relates to other recent papers that study linear or convex programs over spaces of measures. \citet{kleiner_majorization_2021} characterize the extreme points of functions that majorize or are majorized by a given function. Their companion paper \citep{kleiner2022extreme} studies the multidimensional variant of this problem, and \citet{nikzad_constrained_majorization_2022} studies \textit{constrained} majorization, with a particular emphasis on its use in mechanism design. \citet{yang_zentefis_monotone_intervals_2024} characterize the extreme points of ``monotone function intervals,'' using them to reduce persuasion/security-design problems to tractable LPs over measures; and \citet{augiasuhe2025} provide the extreme points of ``convex function intervals,'' with applications to persuasion, delegation, and screening. \citet{yangyang_multidim_monotonicity_2025} look at the extreme points of multidimensional monotone functions, and put their characterization to good use in mechanism design.

\section{The Model}\label{sec:themodel}

We observe two \textbf{Snapshots} around a \emph{specific} (per-unit) tax:
\[\left(p_0,Q_0,\tau_0\right)\quad\text{and}\quad \left(p_1,Q_1,\tau_1\right)\text{,}\] with
\[\Delta p \coloneqq p_1-p_0>0, \quad \Delta Q \coloneqq Q_1-Q_0<0,\quad \text{and} \quad \Delta\tau \coloneqq \tau_1-\tau_0>0\text{.}
\]
We seek \textbf{Bounds} on
\[
\Delta \mathrm{CS},\quad \Delta \mathrm{PS},\quad \Delta \mathrm{TS} \coloneqq \Delta \mathrm{CS}+\Delta \mathrm{PS}+\Delta \mathrm{TR},\quad \text{and} \quad \mathrm{DWL} \coloneqq -\Delta \mathrm{TS},
\]
for
\begin{enumerate}[label={(\roman*)},noitemsep,topsep=0pt]
\item a particular \textbf{Demand-shape class} between \(\left(p_0,Q_0\right)\) and \(\left(p_1,Q_1\right)\),
\item an interval of feasible \textbf{Pass-through} parameters, and
\item an interval of feasible \textbf{Conduct} (market-power) parameters.
\end{enumerate}

\paragraph{Demand.}
We assume that inverse demand \(p=P(Q)\) is strictly decreasing on the relevant range, with demand denoted \(Q(\cdot)=P^{-1}(\cdot)\). The change in (Marshallian) consumer surplus is equal to the area below the demand curve between \(p_0\) and \(p_1\):
\[
\Delta \mathrm{CS} = -\int_{p_0}^{p_1} Q(p) dp\text{.}
\]

\paragraph{Supply.}
Let \(p \colon \left[\tau_0,\tau_1\right]\to\left[p_0,p_1\right]\) be the (unobserved) equilibrium price as a function of the policy level (the specific tax). Basically, for every tax level between the two we observed, imagine the equilibrium price that would prevail holding everything else fixed, then collect those prices into a (regular) function \(p(\tau)\). We don’t observe it, so we will consider all such functions consistent with the endpoints and some mild other conditions and take the best and worst cases over that set when computing welfare bounds.
\begin{assumption}\label{ass:invpt}
\(p\) is absolutely continuous: the pass-through \(\rho(\tau) \coloneqq dp/d\tau\) exists a.e., and, moreover, satisfies
\[0 < \rho_L \leq \rho(\tau) \leq \rho_U<\infty \quad \text{a.e. on } \left[\tau_0,\tau_1\right]\text{.}\]
\end{assumption}
Absolute continuity guarantees \(p(\tau)=p_0+\int_{\tau_0}^{\tau}p'(s)ds\). The lower bound ensures \(p\) is a.e. increasing and (a.e.) invertible. We assume, furthermore, that supply-side behavior can be succinctly captured by the following equation, with assumed bounds on the feasible conduct parameters, i.e., the market structure.
\begin{assumption}\label{ass:conduct}
For all \(Q\in\left[Q_1,Q_0\right]\),
\[\label{eq:conduct}\tag{\(1\)}
p-\tau - \mathrm{MC}(Q) = \kappa(Q)\left[-P'(Q) Q\right],\qquad \text{with} \qquad \kappa(Q)\in\left[\kappa_L,\kappa_U\right]\subseteq[0,1]\text{.}\]
\end{assumption}
This formula originates in \citet{bresnahan1982oligopoly} and is also used and discussed in, e.g., \citet{genesovemullin1998}, \citet{corts1999conduct}, \citet{weylfabinger2013}, and \citet{matsumura2023resolving}.
It nests perfect competition (\(\kappa = 0\)), symmetric \(n\)-firm Cournot competition (\(\kappa \approx 1/n\)), and monopoly/collusion (\(\kappa = 1\)). Moreover, under this conduct identity, the Lerner index satisfies
\[L(p) = \frac{p-\tau-\mathrm{MC}(Q)}{p} = \frac{\kappa(Q)}{\left|\eta(p)\right|}\text{,}
\]
where \(\eta(p)\) is the price-elasticity of demand at price \(p\). To see this, observe that
\[
\frac{dQ}{dp} = \frac{1}{P'(Q)} \quad \Rightarrow \quad 
\left|\eta(p)\right| = -\frac{1}{P'(Q)} \frac{p}{Q}
\quad \Rightarrow \quad
- P'(Q) Q = \frac{p}{\left|\eta(p)\right|}\text{.}
\]
Then, simply plug this into \eqref{eq:conduct} and divide by \(p\).

In more detail, the term \(-P'(Q) Q\) is the familiar ``marginal-revenue correction'' from market power: with downward-sloping demand, marginal revenue equals \(p + P'(Q) Q\).
Rewriting \(-P'(Q) Q\) as \(p/\left|\eta\right|\) puts this correction on the elasticity scale at the observed \((p,Q)\).
The conduct identity says that the per-unit net-of-tax markup \(p-\tau-\mathrm{MC}(Q)\) equals a conduct factor \(\kappa(Q)\) times that \(\mathrm{MR}\) correction.
Normalizing by the consumer price \(p\) produces a normalized markup index:
\[
\underbrace{\frac{p-\tau-\mathrm{MC}(Q)}{p}}_{\text{markup as a share of price}}
=
\underbrace{\kappa(Q)}_{\text{market power / conduct}}
\cdot
\underbrace{\frac{1}{\left|\eta(p)\right|}}_{\text{inverse elasticity}}\text{.}
\]
Thus, higher conduct (more market power) or lower elasticity (more inelastic demand) both raise the Lerner index.

We highlight the following special cases:
\begin{itemize}[noitemsep,topsep=0pt]
\item \textbf{Perfect competition} (\(\kappa \equiv 0\)): \(L\equiv 0\) so that \(p=\tau+\mathrm{MC}(Q)\).
\item Monopoly \((\kappa\equiv 1)\): \(L=1/|\eta|\) under the gross-price normalization used here.
%\item \textbf{Monopoly} (\(\kappa \equiv 1\)): \(L=1/\left|\eta\right|\), which is the textbook Lerner rule (here for net-of-tax price).
\item \textbf{Symmetric Cournot with \(n\) firms}: \(\kappa \equiv 1/N\) so that \(L=(1/N)/\left|\eta\right|\). 
More generally, with heterogeneous shares the conduct proxy equals the Herfindahl \(H\), yielding \(L=H/\left|\eta\right|\).
\end{itemize}

\subsection{Preliminary Results}\label{sec:prelim}

We begin by stating and reformulating the central objects of our exercise. In particular, we single out the crucial decomposition of producer surplus that engenders our bounds. For a finite specific-tax change,
\[\label{eq:rev}\tag{\(2\)}
\Delta \mathrm{TR}=\tau_1 Q_1-\tau_0 Q_0\text{.}\]

Next, define 
\[
u(p) \coloneqq \frac{d\tau}{dp}=\frac{1}{\rho(\tau(p))}\text{.}\]
By \Cref{ass:invpt} (and the Fundamental Theorem of Calculus)
\[
u(p)\in[\alpha,\beta] \coloneqq \left[\frac{1}{\rho_U},\frac{1}{\rho_L}\right]\quad\text{a.e. on }\left[p_0,p_1\right],
\quad \text{and} \quad
\int_{p_0}^{p_1} u(p) dp=\Delta\tau\text{.}\]
As we will shortly derive, one important component of the change in producer surplus is the cumulative, direct hit to producers from a tax increase, \(\int_{\tau_0}^{\tau_1} Q(p(\tau)) d\tau\). It can conveniently be written as an integral over prices instead, using the inverse pass-through to change scale:
\begin{remark}\label{rem:cov}
\[
\int_{\tau_0}^{\tau_1} Q(p(\tau)) d\tau = \int_{p_0}^{p_1} Q(p) u(p) dp\text{.}\]
\end{remark}
Toward deriving our formula for the change in producer surplus, let profits at policy \(\tau\) be
\[
\Pi(\tau)=(p(\tau)-\tau) Q(\tau) - \int_0^{Q(\tau)} \mathrm{MC}(q) dq\text{.}
\]
The following lemma states our central observation noted in the introduction. Namely, pass-through affects producer surplus only through a single integral term, the ``inverse-pass-through term'' in \eqref{eq:PSmaster}.
\begin{lemma}\label{lem:PS}
Under \Cref{ass:invpt,ass:conduct},
\[\label{eq:PSmaster}\tag{\(3\)}\Delta \mathrm{PS}
= \underbrace{\int_{p_0}^{p_1} Q(p) dp}_{-\Delta\text{CS}}
 - \underbrace{\int_{p_0}^{p_1} Q(p) u(p) dp}_{\text{inverse-pass-through term}}
 - \underbrace{\int_{p_0}^{p_1}\kappa(Q) Q(p) dp}_{\text{conduct/markup}}\text{.}\]
\end{lemma}

Think of a per-unit tax change as moving along a curve in \((p,\tau)\)-space. Producer surplus moves for three distinct reasons, and \Cref{lem:PS} simply tallies them. First, when price moves by \(dp\), holding the instant's quantity \(Q\) fixed, producers gain the consumer-to-producer transfer \(Qdp\) accumulated as price rises from \(p_0\) to \(p_1\), \(\int Q dp=-\Delta \mathrm{CS}\). Second, the tax itself takes a bite: each instant you levy \(d\tau\) on \(Q\) units, so producers lose \(\int Q d\tau\) (which we rewrite using \Cref{rem:cov}). Third, as \(p\) changes, firms re-optimize quantity along demand: the value of the marginal units that disappear equals the markup on the margin times \(dQ\). Using the conduct identity \eqref{eq:conduct}, this contribution is \(-\int \kappa(Q) Q dp\). Intuitively, stronger conduct \(\kappa\) absorbs more of the potential transfer from higher prices because output contracts more.

We finish this section with some accounting:
\[\Delta \mathrm{TS}_{\mathrm{priv}} \coloneqq \Delta \mathrm{CS}+\Delta \mathrm{PS}
= -\int_{p_0}^{p_1} Q u dp - \int_{p_0}^{p_1}\kappa(Q) Q dp\text{,}\]
\[\Delta \mathrm{TS} \coloneqq \Delta \mathrm{CS}+\Delta \mathrm{PS}+\Delta \mathrm{TR}
= \Delta \mathrm{TR} - \int_{p_0}^{p_1} Q u dp - \int_{p_0}^{p_1}\kappa(Q) Q dp\text{,}\]
and
\[\mathrm{DWL} \coloneqq  -\Delta \mathrm{TS}
= \int_{p_0}^{p_1} Q u dp + \int_{p_0}^{p_1}\kappa(Q) Q dp - \Delta \mathrm{TR}\text{.}\]
Note that \(\Delta \mathrm{TR}\) is given exactly by \eqref{eq:rev}, \(\Delta \mathrm{CS}\) depends only on the curve \(Q(\cdot)\), and the only dependence on inverse pass-through is the term \(\int Q u dp\).

It is straightforward to rewrite the change in producer surplus as 
\[
\Delta \mathrm{PS} = \int_{p_0}^{p_1} Q(p) dp  -  \int_{p_0}^{p_1} Q(p) u(p) dp  - \int_{p_0}^{p_1} L(p) \left|\eta(p)\right| Q(p) dp\text{,}
\]
since \(L(p)\left|\eta(p)\right|=\kappa(Q)\) pointwise.
This interchangeability is useful when one has external information on Lerner indices or demand elasticities instead of direct bounds on \(\kappa(\cdot)\) (say).

\section{The Supply-Side}\label{sec:supply}

Fix any (weakly) decreasing and integrable demand curve \(Q \colon \left[p_0,p_1\right]\to\mathbb{R}_+\). From the decomposition in \eqref{eq:PSmaster}, we see that the only contribution of \(u\) to any of the welfare objects of interest is the term \(\int Q u dp\). Accordingly, it suffices to compute bounds only by manipulating this term.

We define
\[
\mathcal{U} \coloneqq \left\{u\in L^\infty\left(\left[p_0,p_1\right]\right) \colon \ \alpha\le u\le\beta \ \text{a.e.},\ \int_{p_0}^{p_1}u dp=\Delta\tau\right\}\text{,}
\]
and assume that \(\Delta\tau \in \left[\alpha\Delta p, \beta\Delta p\right]\) (or else \(\mathcal U=\varnothing\)). Note that if \(\alpha=\beta\), then \(\mathcal U\neq\emptyset\) if and only if \(\Delta\tau=\alpha\Delta p\), in which case \(u=\alpha\) a.e. must be the inverse pass-through function. In the remainder of this section, we assume \(\alpha<\beta\).

Define \[F(u) \coloneqq \int_{p_0}^{p_1} Q(p) u(p) dp\text{.}\] Our first result identifies the simple form extremal, and, thus, bound-producing inverse pass-through functions take. Define
\[
h\coloneqq\frac{\Delta\tau-\alpha\Delta p}{\beta-\alpha}\in\left[0,\Delta p\right],
\]
so that
\[
\beta \left(p^*-p_0\right)+\alpha\left(p_1-p^*\right)=\Delta\tau
\qquad \Longleftrightarrow \qquad
p^* = p_0+h\text{,}
\]
and
\[
\alpha \left(p^*-p_0\right)+\beta\left(p_1-p^*\right)=\Delta\tau
\qquad \Longleftrightarrow \qquad
p^* = p_1 - h\text{.}
\]
\begin{definition}
We say that \(u\) is \emph{\(\beta\)-\textbf{B}oundary-value \textbf{A}t a \textbf{T}hreshold \textbf{S}witch (\(u\) is \(\beta\)-BATS)} if, on \(\left[p_0,p_1\right]\), \(u\) equals \(\beta\) up to threshold \(p_0 + h\) and then equals \(\alpha\) thereafter. In turn, we say that \(u\) is \emph{\(\alpha\)-BATS} if it equals \(\alpha\) up to \(p_1-h\), then \(\beta\) thereafter.
\end{definition}

\begin{proposition}\label{thm:invpt-extremal}
\(F\) is maximized (minimized) over \(\mathcal U\) by the \(\beta\)-BATS (\(\alpha\)-BATS) \(u\).
\end{proposition}
That is,
\[
u_{\max}(p)=
\begin{cases}
\beta,&p\in\left[p_0,p_0+h\right),\\
\alpha,&p\in\left[p_0+h,p_1\right],
\end{cases}
\qquad
\text{and} \qquad
u_{\min}(p)=
\begin{cases}
\alpha,&p\in\left[p_0,p_1-h\right),\\
\beta,&p\in\left[p_1-h,p_1\right].
\end{cases}\]
Of course, if \(Q\) has flat portions, other optimizers may exist.

Hence, writing \[\tag{\(4\)}\label{eq4}I^{\omega}_1 \coloneqq \int_{p_0}^{p_{\omega}^*}Q(p) dp \quad \text{and} \quad I^{\omega}_2 \coloneqq \int_{p_{\omega}^*}^{p_1}Q(p) dp \quad \left(\omega \in \left\{\min, \max\right\}\right)\text{,}\]we have
\[
\int_{p_0}^{p_1} Q u_{\max} dp=\beta I^{\max}_1+\alpha I^{\max}_2,\quad \text{and} \quad
\int_{p_0}^{p_1} Q u_{\min} dp=\alpha I^{\min}_1+\beta I^{\min}_2\text{.}
\]

The proof is a rearrangement argument, and these extrema take a simple form: \(u(p)\) travels between a high level \(\beta\) and a low level \(\alpha\) along the price axis \(p\in\left[p_0, p_1\right]\). The budget constraint \(\int u=\Delta\tau\) constrains how much ``high time'' we are allowed, but not where to place it. But this is easy, when we are maximizing the objective \(F\), its form rewards allocating the high level where the demand \(Q\) is large. Since demand is (weakly) decreasing, the best use of the scarce ``high time'' is to allocate it to the left (low \(p\)), where \(Q\) is highest. This is the bathtub principle: if we pour a fixed volume of water into a tub whose bottom slopes downward from left to right (here \(Q\) is larger on the left), we fill the deep end first. The result is a single waterline (a cutoff \(p^{*}\)): to the left of \(p^{*}\) we are at \(\beta\), to the right we are at \(\alpha\), and the budget pins down \(p^{*}\). The mirrored logic applies to the minimization task.

Moreover, if an optimizer is not of this threshold form, we can improve it by a tiny swap. Suppose there is a left location with higher weight \(Q_{\mathrm{left}}\) where \(u\) currently equals \(\alpha\), and a right location with lower weight \(Q_{\mathrm{right}}\) where \(u\) equals \(\beta\). Let's move an \(\varepsilon\)-length slice of \(\beta\)-mass from the right to the left and--as we need to keep \(\int u\) fixed--move an \(\varepsilon\)-slice of \(\alpha\)-mass from the left to the right. The objective changes by
\[
\Delta F = \varepsilon(\beta-\alpha)\left(Q_{\mathrm{left}}-Q_{\mathrm{right}}\right) \ge 0\text{,}
\]
strictly so whenever \(Q_{\mathrm{left}}>Q_{\mathrm{right}}\). Iterating these swaps pushes all \(\beta\)-mass leftward and all \(\alpha\)-mass rightward until a single cutoff remains. If \(Q\) is flat at the threshold, swaps within that plateau neither help nor hurt, explaining the possible non-uniqueness there (if \(Q\) is strictly decreasing a.e., the cutoff is unique).

Plugging \(u_{\min}\) into the producer-surplus formula produces
\[\Delta \mathrm{PS} = I^{\min}_1 (1-\alpha) + I^{\min}_2 (1-\beta) - \int \kappa(Q) Qdp\text{.}\]
Furthermore, \(\kappa(Q)\) lies in the interval \(\left[\kappa_L,\kappa_U\right]\) and \(Q(p) \geq 0\) pointwise, hence, the optimizing \(\kappa\) is constant, so the upper \(\mathrm{PS}\) bound is 
\[\Delta \mathrm{PS} =(1-\kappa_L-\alpha)I^{\min}_1+(1-\kappa_L-\beta)I^{\min}_2\text{.}\]
In turn, the lower bound exchanges \(\alpha\) and \(\beta\) and \(\kappa_L\) with \(\kappa_U\):
\[\Delta \mathrm{PS} =(1-\kappa_U-\beta)I^{\max}_1+(1-\kappa_U-\alpha)I^{\max}_2\text{.}\]

Thus, we have
\[-(\beta I^{\max}_1+\alpha I^{\max}_2)-\kappa_U \left(I^{\max}_1 + I^{\max}_2\right) \leq \Delta \mathrm{TS}_{\mathrm{priv}} \leq -(\alpha I^{\min}_1+\beta I^{\min}_2)-\kappa_L \left(I^{\min}_1 + I^{\min}_2\right)\text{,}\]
and
\[(\alpha I^{\min}_1+\beta I^{\min}_2)+ \kappa_L \left(I^{\min}_1 + I^{\min}_2\right) -\Delta \mathrm{TR} \leq \mathrm{DWL} \leq (\beta I^{\max}_1+\alpha I^{\max}_2)+\kappa_U \left(I^{\max}_1 + I^{\max}_2\right)-\Delta \mathrm{TR}\text{.}\]

\section{The Demand-Side Results}\label{sec:demand}

In our demand-side analysis, we mirror our supply-side agnosticism. Following KV, we do not make parametric assumptions about the true demand curve, but instead study the collection of demand curves that satisfy various shape restrictions. Namely, we posit two strictly-increasing and continuously differentiable functions \(A(q)\) and \(B(p)\) (that satisfy further properties to be specified shortly) and define a function \(G\coloneqq A\circ Q\) as
\[
G\left(s\right)\coloneqq A\left(Q\left(p\left(s\right)\right)\right),\quad \text{for }s\in\left[s_{0},s_{1}\right]=\left[B\left(p_{0}\right),B\left(p_{1}\right)\right].
\]
Write \(Q_{0}\coloneqq Q\left(p_{0}\right)\) and \(Q_{1}\coloneqq Q\left(p_{1}\right)\).

Economically, \(B\) picks the ``units'' along the price axis (levels versus logs, etc.), while \(A\) picks the ``units'' along the quantity axis (levels, logs, etc.). We then impose shape restrictions on \(G\) by assuming that the derivative of \(G\) lies in a specified interval. Importantly, \(A\) and \(B\) are simply ways of reparametrizing demand so that shape information is stated in the most natural units.

The researcher, thus, picks \(A\) and \(B\) to align the constraint with the empirical object she trusts: if she wants to make assumptions on the slope of demand, she posits \(A\left(q\right)=q\) and \(B\left(p\right)=p\), so that \(G'\left(s\right)=dQ/dp\). If she wants to assume an interval of elasticities, \(A\left(q\right)=\log q\) and \(B\left(p\right)=\log p\), so that \(G'\left(s\right)=d\log Q/d\log p\). For a semi-elasticity (percent change in quantity for a change in price) interval, \(A\left(q\right)=\log q\), and \(B\left(p\right)=p\), so that \(G'\left(s\right)=d\log Q/dp\); and for the opposite (change in quantity for a percent change in price), \(A\left(q\right)=q\) and \(B\left(p\right)=\log p\), so that \(G'\left(s\right)=dQ/d\log p\). The point is that researchers typically have outside evidence or priors about one of these objects (e.g., an elasticity range from prior studies, or a local slope range from IV).

\subsection{Standing Assumptions}

We impose the following additional technical assumptions concerning \(A\) and \(B\). Write \(Q_i\coloneqq Q\left(p_i\right)\) for \(i\in\left\{0,1\right\}\), and let \(\mathcal Q\coloneqq\left[Q_1,Q_0\right] \subset \mathbb{R}_{++}\) (with \(Q_1 < Q_0\)). We specify that \(A\) is \(C^1\) and strictly increasing on an open interval containing \(\mathcal Q\), with \(A'\left(q\right)>0\) throughout that interval; and \(B\) is \(C^1\) and strictly increasing on an open interval containing \(\left[p_0,p_1\right]\), with \(B'\left(p\right)>0\) throughout that interval.

\subsection{Demand-Side Analysis}

Following KV, we fix a pair \(A\) and \(B\). KV characterize extremal demand curves in two shape classes. In contrast, our main analysis specializes to their gradient-interval class (Assumption 1 in their paper).

\begin{assumption}\label{ass:kv-gradient}
Given constants \(\underline f\le\overline f\le0\), the gradient class consists of demand curves \(Q\colon\left[p_0,p_1\right] \to \mathcal{Q}\) such that \(G\) is absolutely continuous on the interval \(\left[s_0,s_1\right]\), \(G\left(s_i\right)=A\left(Q_i\right)\) for \(i\in\left\{0,1\right\}\), and
\(\underline f\le G'(s)\le\overline f
\quad\text{a.e. on }\left[s_0,s_1\right]\).
\end{assumption}
Equivalently, in \(p\)-space, this last restriction can be written a.e. as
\[G'\left(B\left(p\right)\right) = \frac{A'\left(Q\left(p\right)\right)Q'\left(p\right)}{B'\left(p\right)}
 \in \left[\underline f,\overline f\right].\]

In our environment, for fixed \(u\) and \(Q\), each object of interest is linear in \(\kappa\) and \(\kappa\) does not change the sets of \(u\)s or \(Q\)s we optimize over. Thus, welfare bounds occur when \(\kappa\) equals one of the endpoints of the feasible interval, and we set \(\kappa(p)=\kappa\in\left[\kappa_L,\kappa_U\right]\) for all \(p\).

From \Cref{thm:invpt-extremal}, the two extremal \(u\)s satisfy
\[
\int_{p_0}^{p_1}Q(p)u_{\max}(p)\,dp
=\alpha\int_{p_0}^{p_1}Q(p)\,dp+\left(\beta-\alpha\right)\int_{p_0}^{p_{\max}^{*}}Q(p)\,dp,
\]
and
\[
\int_{p_0}^{p_1}Q(p)u_{\min}(p)\,dp
=\alpha\int_{p_0}^{p_1}Q(p)\,dp+\left(\beta-\alpha\right)\int_{p_{\min}^{*}}^{p_1}Q(p)\,dp.
\]
Thus, after restricting \(u\) to one extremal branch \(\omega\in\left\{\max,\min\right\}\), any welfare object listed above can be written as
\[
W_\omega(Q)=C_0+a_0\int_{p_0}^{p_1}Q(p)\,dp+a_1\int_{K_\omega}Q(p)\,dp,
\]
where
\[K_{\max}\coloneqq\left[p_0,p_{\max}^{*}\right] \qquad \text{and} \qquad K_{\min}\coloneqq\left[p_{\min}^{*},p_1\right].\]
The coefficients are
\[\begin{array}{c|ccc}
W&C_0&a_0&a_1\\
\hline
\Delta\mathrm{TR}&\Delta\mathrm{TR}&0&0\\
\Delta\mathrm{CS}&0&-1&0\\
\Delta\mathrm{PS}&0&1-\alpha-\kappa&\alpha-\beta\\
\Delta\mathrm{TS}_{\mathrm{priv}}&0&-\alpha-\kappa&\alpha-\beta\\
\Delta\mathrm{TS}&\Delta\mathrm{TR}&-\alpha-\kappa&\alpha-\beta\\
\mathrm{DWL}&-\Delta\mathrm{TR}&\alpha+\kappa&\beta-\alpha
\end{array}\]
where \(C_0\) does not depend on the demand curve.

Next, as \(dp=ds/B'\left(p(s)\right)\) and \(Q(p)=A^{-1}\left(G\left(B(p)\right)\right)\),
\[
\int_{p_0}^{p_1}Q(p)\,dp=\int_{s_0}^{s_1}\frac{1}{B'\left(p(s)\right)}A^{-1}\left(G(s)\right) ds.
\]
Writing \(s_\omega^{*}\coloneqq B\left(p_\omega^{*}\right)\), we have
\[
B\left(K_{\max}\right)=\left[s_0,s_{\max}^{*}\right],
\qquad \text{and} \qquad
B\left(K_{\min}\right)=\left[s_{\min}^{*},s_1\right].
\]
Thus, for \(\omega\in\left\{\max,\min\right\}\),
\[
W_\omega=C_0+\int_{s_0}^{s_1}\varphi_\omega(s)A^{-1}\left(G(s)\right) ds,
\qquad \text{and} \qquad
\varphi_\omega(s)\coloneqq\frac{a_0+a_1\mathbf 1_{B\left(K_\omega\right)}(s)}{B'\left(p(s)\right)}.
\]
For the remainder of this subsection, we fix a welfare object, \(\omega\), and \(\kappa\), and write \(\varphi\) for \(\varphi_\omega\).

Crucially, \(\varphi\) changes sign at most once.\footnote{If the welfare measure is \(\Delta PS\), then \(\varphi\) changes sign if and only if \(1-\beta<\kappa<1-\alpha\), in which case \(\varphi>0\) on \(B(K_\omega)^c\) and \(\varphi<0\) on \(B(K_\omega)\). At the boundary values \(\kappa=1-\beta\) or \(\kappa=1-\alpha\), \(\varphi\) is zero on one of these regions and has a constant sign on the other. These weakly one-signed cases are covered by the one-signed argument in \Cref{rem:onesigned}.} If \(\varphi\) does not change sign, we can appeal to Proposition 1 in \cite{KangVasserman2025}, which states that the extreme demand curves are pointwise maximal or minimal, and are therefore optimal here as well. Naturally, for \(\varphi\le0\) a.e., we multiply the objective by \(-1\) and swap \(\max\) and \(\min\).

\begin{remark}\label{rem:onesigned}
Fix a welfare object, a branch \(\omega\), and \(\kappa\). If \(\varphi\) does not change sign, welfare bounds are given by the corresponding KV extremal demand curves.
\end{remark}

%When \(\varphi\) changes sign, we restrict attention to Assumption \ref{ass:kv-gradient}. 
Let \(L\coloneqq s_1-s_0\) and \(\Delta A\coloneqq A\left(Q_1\right)-A\left(Q_0\right)\). For \(f\in L^\infty\left(\left[s_0,s_1\right]\right)\),\footnote{The space \(L^\infty\left(\left[s_0,s_1\right]\right)\) is the space of essentially bounded measurable functions on \(\left[s_0,s_1\right]\).} we define
\[
G_f(s)\coloneqq A\left(Q_0\right)+\int_{s_0}^{s}f(t)\,dt,
\qquad \text{and} \qquad
I(f)\coloneqq\int_{s_0}^{s_1}\varphi(s)A^{-1}\left(G_f(s)\right) ds,
\]
and
\[
\mathcal F\coloneqq\left\{f\in L^\infty\left(\left[s_0,s_1\right]\right)\colon\underline f\le f\le\overline f\ \text{a.e.},\ G_f(s_1)=A\left(Q_1\right)\right\}.
\]
We assume that \(\mathcal F\neq\emptyset\); equivalently, \(\underline fL\le\Delta A\le\overline fL\). Moreover, if \(\underline f=\overline f\), the result is immediate. Consequently, we assume that \(\underline f<\overline f\) and set
\[m\coloneqq\frac{\Delta A-\underline fL}{\overline f-\underline f}\in\left[0,L\right].\]

\begin{definition}
    We say that \(\varphi\) is \textit{single-crossing} if \(\varphi\in L^1\left(\left[s_0,s_1\right]\right)\), \(\varphi\neq0\) a.e., and there exists \(s_\varphi\in\left(s_0,s_1\right)\) such that either
\[\varphi>0\ \text{a.e. on }\left[s_0,s_\varphi\right),
\qquad \text{and} \qquad
\varphi<0\ \text{a.e. on }\left(s_\varphi,s_1\right],\]
or the reverse.
\end{definition}

For \(f\in\mathcal F\), write
\[
S_f\coloneqq\left\{s\in\left[s_0,s_1\right]\colon f(s)=\overline f\right\}.
\]
\begin{definition}
    We say that \(f\) \textit{takes \textbf{E}ndpoints with \textbf{A}t most \textbf{T}wo \textbf{S}witches}, \(f\) ``EATS,'' if there exist \(t_1,t_2\in\left[s_0,s_1\right]\), with \(t_1\le t_2\), and \(c_0,c_1,c_2\in\left\{\underline f,\overline f\right\}\) such that, a.e.,
\[
f(s)=c_0\ \text{on }\left[s_0,t_1\right],
\qquad
f(s)=c_1\ \text{on }\left(t_1,t_2\right),
\qquad \text{and} \qquad
f(s)=c_2\ \text{on }\left[t_2,s_1\right].
\]
\end{definition}

\begin{proposition}\label{prop:ass3}
If  \(\varphi\) is single-crossing, then both extrema of \(I\) over \(\mathcal F\) are attained by some \(f^*\) that EATS, with \(\lambda\left(S_{f^*}\right)=m\).
\end{proposition}

The proof of \Cref{prop:ass3} is in the appendix. Note that KV also conduct their exercise when there are curvature restrictions on \(G\), under the additional smoothness needed for that class. When the welfare weight \(\varphi\) is one-signed, we can port their curvature-class extremal curves directly, as the welfare objective is monotone (in the pointwise sense). The sign-changing case is different, and we do not tackle this challenging problem.

The proof begins by separating the choice of the demand curve into two pieces. The value that links the left and right sides is the value of the transformed demand curve at the point where the weight \(\varphi\) changes sign. We fix this value and write \(\gamma\coloneqq G_f\left(s_\varphi\right)\). Importantly, once \(\gamma\) is fixed, the left side of the curve only has to connect \(\left(s_0,A\left(Q_0\right)\right)\) to \(\left(s_\varphi,\gamma\right)\), while the right side only has to connect \(\left(s_\varphi,\gamma\right)\) to \(\left(s_1,A\left(Q_1\right)\right)\), always respecting the same derivative bounds. We prove a lemma (\Cref{lem:grad-envelope}), which tells us exactly which curve (satisfying our other restrictions) is highest and which is lowest on each of these two subintervals.

This comparison is useful because the sign of \(\varphi\) tells us whether we want \(G\) to be high or low. Since \(A^{-1}\) is increasing, raising \(G(s)\) raises the quantity term \(A^{-1}\left(G(s)\right)\). Therefore, where \(\varphi\) is positive, the objective is improved by taking the highest feasible curve, and where \(\varphi\) is negative, the objective is improved by taking the lowest one. Thus, conditional on \(\gamma\), the optimal curve is obtained by gluing together the appropriate upper and lower envelopes from the lemma (\Cref{lem:grad-envelope}).

The glued curve has a simple derivative pattern. If \(\varphi\) is positive on the left and negative on the right, the curve uses the upper envelope on the left and the lower envelope on the right, so its derivative follows the pattern \(\overline f,\underline f,\overline f\), up to degenerate pieces and null sets. If the signs are reversed, the derivative follows the pattern \(\underline f,\overline f,\underline f\), again up to degenerate pieces and null sets. In either case, the derivative uses only the endpoint values and switches at most twice.

It remains to choose the best value of \(\gamma\). The feasible values of \(\gamma\) form a nonempty compact set \(\Gamma\), and the value of the objective along the constructed curve varies continuously with \(\gamma\). In fact, the constructed objective is Lipschitz in \(\gamma\), because the envelopes are Lipschitz in \(\gamma\) and \(A^{-1}\) is Lipschitz on the relevant compact range. Therefore, some \(\gamma^*\in\Gamma\) attains the maximum. The corresponding derivative \(f^*\coloneqq f_{\gamma^*}\) EATS. Finally, because \(f^*\) only takes the two values \(\underline f\) and \(\overline f\), and because it must satisfy the endpoint condition \(\int_{s_0}^{s_1}f^*(s)ds=\Delta A\), the measure of the set on which \(f^*=\overline f\) is necessarily \(m\). The minimization statement follows by applying the same argument to \(-\varphi\).

\section{The Main Result}\label{sec:mainresult}

\Cref{thm:invpt-extremal} identifies the bound-producing form the inverse pass-through function \(u\) takes, \textit{for a given demand curve \(Q\)}. \secref{sec:demand} then characterizes the corresponding demand-side extrema: when the induced welfare weight is one-signed, the KV extremal curves deliver the bounds; and when the weight is single-crossing, \Cref{prop:ass3} shows that a curve that EATS is extremal. It remains to combine these findings. Because the supply-side cutoffs are pinned down by the endpoints and pass-through bounds, and do not depend on the particular demand curve beyond monotonicity, we can first restrict attention to the two extremal inverse pass-through functions and then solve the corresponding demand-side problem:
\begin{theorem}\label{thm:fourcorners}
Global welfare bounds are attained by candidates with the following extremal structure:
\begin{enumerate}[label={(\roman*)},noitemsep,topsep=0pt]
\item \(u\) takes the \Cref{thm:invpt-extremal} form, i.e., \(u\) is either \(\beta\)- or \(\alpha\)-BATS; and
\item conditional on this \(u\), the welfare object, and \(\kappa\), \(Q\) corresponds to an \(f\) that EATS.
\end{enumerate}
Moreover, \(\kappa\) is also extremal: \(\kappa\in\left\{\kappa_L,\kappa_U\right\}\).
\end{theorem}
\begin{proof}
Take a welfare object. For demand curve \(Q\) satisfying Assumption 3 and conduct value \(\kappa\), the only term depending on \(u\) is a scalar multiple of \(\int_{p_0}^{p_1}Q(p)u(p)dp\). \Cref{thm:invpt-extremal}, therefore, implies that, conditional on \(Q\) and \(\kappa\), an extremum over \(u\) is attained by one of the two BATS inverse pass-through functions, \(u_{\max}\) or \(u_{\min}\). Conditional on either such \(u\), and on \(\kappa\), the welfare object can be written as
\[
W_\omega(Q)=C_0+\int_{s_0}^{s_1}\varphi_\omega(s)A^{-1}(G(s))ds,
\]
as in \secref{sec:demand}. If \(\varphi_\omega\) is one-signed, the KV extremal curve solves the demand-side problem; whereas if \(\varphi_\omega\) is single-crossing, \Cref{prop:ass3} delivers an extremum that EATS. Hence, an extremum is attained by a candidate with \(u\) BATS and \(Q\) corresponding to an \(f\) that EATS. Finally, for every fixed \((Q,u)\), the welfare expressions are linear in \(\kappa\). Thus, whenever the welfare object depends on \(\kappa\), an extremum over \([\kappa_L,\kappa_U]\) is attained at \(\kappa_L\) or \(\kappa_U\).
\end{proof}

Given any selected candidate \((Q,u,\kappa)\), define the implied tax schedule by
\[
\tau_u(p)\coloneqq\tau_0+\int_{p_0}^p u(r)dr.
\]
Where \(P'\) exists, the implied marginal-cost curve is
\[
\mathrm{MC}(Q(p))\coloneqq p-\tau_u(p)-\kappa\left[-P'(Q(p))Q(p)\right].
\]

\subsection{A Robust-Bounds ``Cookbook''}

You, the analyst, observe two equilibria \((p_0,Q_0,\tau_0)\) and \((p_1,Q_1,\tau_1)\) with \(p_1>p_0\), \(Q_1<Q_0\), \(\Delta p=p_1-p_0\), and \(\Delta\tau=\tau_1-\tau_0>0\).
You posit a pass-through interval \(\rho_L\le\rho\le\rho_U\) and a conduct interval \(\kappa_L\le\kappa\le\kappa_U\).
Recall the inverse pass-through bounds
\[
\alpha\coloneqq\frac{1}{\rho_U},\qquad \beta\coloneqq\frac{1}{\rho_L},\quad \text{and} \quad \alpha\le\beta\text{.}
\]

\begin{enumerate}
\item \textbf{Compute the extremal inverse pass-through functions.} \Cref{thm:invpt-extremal} gives the two extremal inverse pass-through functions
\[
u_{\max}(p)=
\begin{cases}
\beta,& p\in\left[p_0,p^{*}_{\max}\right),\\
\alpha,& p\in\left[p^{*}_{\max},p_1\right],
\end{cases}
\quad \text{and} \quad
u_{\min}(p)=
\begin{cases}
\alpha,& p\in\left[p_0,p^{*}_{\min}\right),\\
\beta,& p\in\left[p^{*}_{\min},p_1\right]\text{,}
\end{cases}
\]
with corresponding cutoffs \(p^{*}_{\max}=p_0 + h\) and \(p_{\min}^{*}= p_1-h\). You will run \emph{both} extremals in Steps~\ref{step:demand}-\ref{step:fourcorners}.

\item\label{step:demand} \textbf{Compute the demand-side extrema.}
For each welfare object, each \(u\in\left\{u_{\max},u_{\min}\right\}\), and each relevant conduct endpoint \(\kappa\in\left\{\kappa_L,\kappa_U\right\}\), form the induced welfare weight \(\varphi\) as in \secref{sec:demand}. If \(\varphi\) is one-signed, use the corresponding KV extremal curve, which is a one-switch EATS curve. If \(\varphi\) is single-crossing, \Cref{prop:ass3} furnishes an EATS extremizer. Use the maximizing or minimizing demand-side extremizer according to the welfare bound being computed.

\item\label{step:three} \textbf{Compute the building blocks for the welfare items.} Compute
\[
\int_{p_0}^{p_1}Q(p)u(p)dp \quad \text{and} \quad \int_{p_0}^{p_1}Q(p)dp\text{,}
\]
for each relevant candidate pair \((Q,u)\) from Steps~1 and~\ref{step:demand}.

\item\label{step:fourcorners} \textbf{Compute welfare pieces.} Recall the formulas
\[
\Delta \mathrm{CS}=-\int_{p_0}^{p_1}Q(p)dp,\qquad \Delta \mathrm{PS}
=-\Delta \mathrm{CS}
-\int_{p_0}^{p_1}Q(p)u(p)dp
-\int_{p_0}^{p_1}\kappa(Q)Q(p)dp\text{,}
\]
\[
\Delta \mathrm{TS}_{\mathrm{priv}}\coloneqq\Delta \mathrm{CS}+\Delta \mathrm{PS}\quad \text{and} \quad -\mathrm{DWL}=\Delta \mathrm{TS}\coloneqq\Delta \mathrm{CS}+\Delta \mathrm{PS}+\Delta \mathrm{TR}\text{.}
\]
Compute each object, or any convex combination thereof, under the relevant candidates. For lower PS/TS and higher DWL, set \(\kappa=\kappa_U\); and for upper PS/TS and lower DWL, set \(\kappa=\kappa_L\). If \(\kappa\) is point-identified, just plug it in.
\end{enumerate}

\bibliography{sample.bib}

\appendix

\section{Mathematical Conventions and External Results}

For an interval \(\left[s_0,s_1\right]\), let \(\lambda\) denote the Lebesgue measure on \(\left[s_0,s_1\right]\).
For \(1\le p<\infty\), \(L^p\left(\left[s_0,s_1\right]\right)\) denotes the space of (equivalence classes of) measurable functions \(g\) such that \(\int_{s_0}^{s_1}|g(s)|^p ds<\infty\).
The space \(L^\infty\left(\left[s_0,s_1\right]\right)\) consists of essentially bounded measurable functions.
Since \(\lambda\left(\left[s_0,s_1\right]\right)<\infty\), we have \(L^\infty\left(\left[s_0,s_1\right]\right)\subset L^1\left(\left[s_0,s_1\right]\right)\).

All equalities and inequalities between functions (e.g., \(\underline f\le f\le\overline f\)) are understood to hold \(\lambda\)-a.e.
For a measurable set \(H\subseteq\left[s_0,s_1\right]\), \(1_H\) denotes its indicator function. We write \(\mathrm{AC}\left(\left[s_0,s_1\right]\right)\) for the absolutely continuous functions on the interval \(\left[s_0,s_1\right]\).
For a measurable \(H\subseteq[s_0,s_1]\), \(\partial H\) denotes its topological boundary.

In the proof of \Cref{thm:invpt-extremal}, we use the following inequality--see, e.g., Proposition 1.6 in \cite{burchard2009short} for a reference. Let \(a\) and \(b\) be generic nonnegative measurable functions \(a,b \colon \left[s_0,s_1\right] \to \left[0,\infty\right]\). We define \(a\)'s distribution function
\[\mu_a(t) \coloneqq \lambda\left(\left\{s \in \left[s_0,s_1\right] \colon a(s) > t\right\}\right) \qquad \forall \  t \ge 0.\]
The \textit{decreasing rearrangement} \(a^{\downarrow}\) is the (a.e. unique) nonincreasing function on \(\left[s_0,s_1\right]\) that is
equimeasurable with \(a\), meaning that
\[\lambda\left(\left\{s \in \left[s_0,s_1\right] \colon a^{\downarrow}(s) > t\right\}\right) = \mu_a(t) \qquad \forall \ t \ge 0.\]
Equivalently, since \(\lambda\left(\left[s_0,s_1\right]\right)=s_1-s_0\), we may write
\[a^{\downarrow}(s) \coloneqq \inf\left\{t \ge 0 \colon \mu_a(t) \le s - s_0\right\} \qquad \forall \ s \in \left[s_0,s_1\right],\]
and define \(b^{\downarrow}\) analogously.

\begin{lemma}[Hardy-Littlewood Rearrangement Inequality]\label{hwrearrange}
    If \(a,b\ge 0\) are measurable, then
\[
\int a b\ \le \int a^\downarrow b^\downarrow,
\]
with equality when \(a,b\) are comonotone.
\end{lemma}

\iffalse
To justify existence of optimizers (Lemmas~\ref{lem:4.3-filippov-cesari} and \ref{lem:4.9-filippov-cesari}), we use the Arzel\`a-Ascoli theorem.
\begin{theorem}[Arzel\`a-Ascoli on an interval]\label{thm:aa-interval}
Fix \(s_0<s_1\) and \(d\in\mathbb N\). Let \(\left\{x_n\right\}\subset C\left(\left[s_0,s_1\right];\mathbb R^d\right)\) be uniformly bounded and equicontinuous on \(\left[s_0,s_1\right]\).
Then there exist a subsequence \(\left\{x_{n_k}\right\}\) and \(x\in C\left(\left[s_0,s_1\right];\mathbb R^d\right)\) such that \(x_{n_k}\to x\) uniformly on \(\left[s_0,s_1\right]\).
\end{theorem}

\begin{corollary}\label{corr:aa-equi-lip}
Fix \(s_0<s_1\) and \(d\in\mathbb N\). Let \(\left\{x_n\right\}\subset C\left(\left[s_0,s_1\right];\mathbb R^d\right)\) satisfy the equi-Lipschitz bound
\[
\left|x_n(s)-x_n(t)\right|\le K|s-t|
\qquad\text{for all }s,t\in\left[s_0,s_1\right]\text{ and all }n,
\]
for some \(K<\infty\), and assume \(\sup_n\left|x_n(s_0)\right|<\infty\).
Then \(\left\{x_n\right\}\) is uniformly bounded and equicontinuous on \(\left[s_0,s_1\right]\), hence, admits a uniformly convergent subsequence by \Cref{thm:aa-interval}.
Any uniform limit \(x\) is \(K\)-Lipschitz and, therefore, belongs to \(AC\left(\left[s_0,s_1\right];\mathbb R^d\right)\).
\end{corollary}

\fi

\section{Omitted Proofs}\label{appendix}

\subsection{\Cref{rem:cov} and \Cref{lem:PS} Proof}

Recall that the remark is the simple equation
\[
\int_{\tau_0}^{\tau_1} Q(p(\tau)) d\tau = \int_{p_0}^{p_1} Q(p) u(p) dp\text{.}\]

\bigskip

\begin{proof}[Proof of \Cref{rem:cov}]
\Cref{ass:invpt} implies \(p\) has an inverse \(\tau(\cdot)\) with \(\tau'(p)=1/p'(\tau(p))=1/\rho(\tau(p))=u(p)\) a.e. Thus, 
\[
\int_{\tau_0}^{\tau_1} Q(p(\tau))d\tau = \int_{p_0}^{p_1} Q(p) \tau^{\prime}(p)dp = \int_{p_0}^{p_1} Q(p)u(p) dp\text{,}\]
yielding the claim.
\end{proof}

\bigskip

Then, \Cref{lem:PS} is the equation
\[\Delta \mathrm{PS}
= \int_{p_0}^{p_1} Q(p) dp
 - \int_{p_0}^{p_1} Q(p) u(p) dp
 - \int_{p_0}^{p_1}\kappa(Q) Q(p) dp\text{.}\]

 \bigskip
 
\begin{proof}[Proof of \Cref{lem:PS}]
We differentiate \(\Pi\):
\[
\frac{d\Pi}{d\tau}=\frac{d}{d\tau}\left[\left(p-\tau\right)Q\right]-\frac{d}{d\tau}\left(\int_0^{Q}MC\right)
=\left(p'(\tau)-1\right)Q+\left(p-\tau\right)Q'(\tau)-\mathrm{MC}(Q)Q'(\tau)\text{.}
\]
We set \(p'(\tau)=\rho\). By \eqref{eq:conduct}, \(p-\tau-\mathrm{MC}(Q)=\kappa(Q)[-P'(Q)Q]\) and by the chain rule \(Q'(\tau)=\frac{dQ}{dp}p'(\tau) =(1/P'(Q))\rho\), so,
\[
\frac{d\Pi}{d\tau}=(\rho-1)Q+\kappa(Q)\left[-P'(Q)Q\right]\frac{\rho}{P'(Q)}=(\rho-1)Q-\kappa(Q) Q \rho\text{.}
\]

We integrate this expression over \(\tau\in\left[\tau_0,\tau_1\right]\), yielding
\[
\Delta \mathrm{PS}=\int_{\tau_0}^{\tau_1}(\rho-1)Q d\tau-\int_{\tau_0}^{\tau_1}\kappa(Q)Q \rho d\tau\text{.}
\]
Finally, we use \(dp=\rho d\tau\) and \Cref{rem:cov} to obtain
\[\int \rho Q d\tau=\int Q dp, \quad \int Q d\tau=\int Q u dp \quad \text{and} \quad \int \kappa(Q)Q \rho d\tau=\int \kappa(Q)Q dp\text{.}\] Plugging these in yields \eqref{eq:PSmaster}.
\end{proof}

\subsection{\Cref{thm:invpt-extremal} Proof}
\begin{proof}[Proof of \Cref{thm:invpt-extremal}]
Recalling that \(\alpha<\beta\), set
\[w(p)\coloneqq\frac{u(p)-\alpha}{\beta-\alpha},\] so that \(0\le w\le1\) a.e. and
\[
\int_{p_0}^{p_1}w(p)dp
=
h\coloneqq\frac{\Delta\tau-\alpha\Delta p}{\beta-\alpha}
\in\left[0,\Delta p\right].
\]
Moreover,
\[
F(u)=\alpha\int_{p_0}^{p_1}Q(p)dp+\left(\beta-\alpha\right)\int_{p_0}^{p_1}Q(p)w(p)dp.
\]
Accordingly, optimizing \(F\) over \(\mathcal U\) is equivalent to optimizing
\[
\Phi(w)\coloneqq\int_{p_0}^{p_1}Q(p)w(p)dp
\]
over all measurable \(w\) satisfying \(0\le w\le1\) a.e. and \(\int w=h\).

We solve the maximization problem--the minimization problem is analogous, applying the same maximization argument to \(1-w\). For \(t\in\left[0,1\right]\), let
\[
E_t\coloneqq\left\{p\in\left[p_0,p_1\right]\colon w(p)>t\right\}.
\]
By the layer-cake representation,
\[
w(p)=\int_0^1\mathbf 1_{E_t}(p)dt
\quad\text{a.e.}
\]
and so, by Tonelli's theorem,
\[
\Phi(w)=\int_0^1\int_{E_t}Q(p)dp\,dt.
\]

For each \(t\), applying \Cref{hwrearrange} on \(\left[p_0,p_1\right]\) to \(Q\) and \(\mathbf 1_{E_t}\), we obtain
\[
\int_{E_t}Q(p)dp
\le
\int_{p_0}^{p_0+\lambda(E_t)}Q(p)dp,
\]
as \(Q\) is weakly decreasing and the decreasing rearrangement of \(\mathbf 1_{E_t}\) is \(\mathbf 1_{\left[p_0,p_0+\lambda(E_t)\right]}\), up to null sets.

For \(x \in \left[0,\Delta p\right]\), we define \(H(x)\coloneqq\int_{p_0}^{p_0+x}Q(p)dp\). Since \(Q\) is weakly decreasing, \(H\) is concave. Therefore, by Jensen's inequality and the layer-cake identity,
\[
\Phi(w)
\le
\int_0^1H\left(\lambda(E_t)\right)dt
\le
H\left(\int_0^1\lambda(E_t)dt\right)
=
H\left(\int_{p_0}^{p_1}w(p)dp\right)
=
H(h).
\]
Thus,
\[
\Phi(w)\le\int_{p_0}^{p_0+h}Q(p)dp,
\]
with equality at \(w=\mathbf 1_{\left[p_0,p_0+h\right)}\).

\iffalse
For the minimization problem, apply the same maximization argument to \(1-w\), which satisfies \(0\le1-w\le1\) a.e. and
\[
\int_{p_0}^{p_1}\left(1-w(p)\right)dp=\Delta p-h.
\]
Then
\[
\int_{p_0}^{p_1}Q(p)\left(1-w(p)\right)dp
\le
\int_{p_0}^{p_1-h}Q(p)dp.
\]
Consequently,
\[
\Phi(w)
=
\int_{p_0}^{p_1}Q(p)dp
-
\int_{p_0}^{p_1}Q(p)\left(1-w(p)\right)dp
\ge
\int_{p_1-h}^{p_1}Q(p)dp,
\]
with equality at \(w=\mathbf 1_{\left[p_1-h,p_1\right]}\).
\fi

Returning to \(u=\alpha+\left(\beta-\alpha\right)w\), the maximizing and minimizing inverse pass-through functions are
\[
u_{\max}(p)=\alpha+\left(\beta-\alpha\right)\mathbf 1_{\left[p_0,p_0+h\right)}(p) \qquad \text{and} \qquad
u_{\min}(p)=\alpha+\left(\beta-\alpha\right)\mathbf 1_{\left[p_1-h,p_1\right]}(p).
\]
Therefore,
\[
p_{\max}^*=p_0+h
=
p_0+\frac{\Delta\tau-\alpha\Delta p}{\beta-\alpha} \quad \text{and} \quad
p_{\min}^*=p_1-h
=
p_1-\frac{\Delta\tau-\alpha\Delta p}{\beta - \alpha}. \qedhere
\]
\end{proof}

\subsection{Proof of \Cref{prop:ass3}}

We say that a quadruple of real numbers \(\left(a,b,g_a,g_b\right)\), with \(a<b\) is \textit{feasible} if
\[
\underline f(b-a)\le g_b-g_a\le\overline f(b-a).
\]
For a feasible quadruple, we define the upper and lower envelopes
\[
\overline G(s)\coloneqq\min\left\{g_a+\overline f(s-a),g_b-\underline f(b-s)\right\} \quad \text{and} \quad \underline G(s)\coloneqq\max\left\{g_a+\underline f(s-a),g_b-\overline f(b-s)\right\}.
\]
%A \textit{feasible curve} for \(\left(a,b,g_a,g_b\right)\) is any \(G\in AC\left(\left[a,b\right]\right)\) satisfying \(G(a)=g_a\), \(G(b)=g_b\), and \(\underline f\le G'\le\overline f\) a.e.

\begin{lemma}\label{lem:grad-envelope}
Fix a feasible quadruple. Then:
\begin{enumerate}
    \item\label{lemmait1} The upper and lower envelopes \(\overline G,\underline G\in AC([a,b])\), and satisfy \(\overline G(a)=\underline G(a)=g_a\), \(\overline G(b)=\underline G(b)=g_b\), and \(\underline G', \overline G' \in \left\{\underline f, \overline f\right\}\) a.e.
    \item\label{lemmait2} For any \(G\in AC([a,b])\) with \(G(a)=g_a\), \(G(b)=g_b\), and \(\underline f\le G'\le\overline f\) a.e., we have \(\underline G(s)\le G(s)\le\overline G(s)\) for all \(s\in[a,b]\).
    \item\label{lemmait3} The upper (lower) envelope \(\overline G'\) (\(\underline G'\)) is equal to \(\overline f\) then \(\underline f\) (\(\underline f\) then \(\overline f\)).
\end{enumerate}
\end{lemma}
\begin{proof}
    \ref{lemmait1} is immediate. For \ref{lemmait2}, take arbitrary \(G\in AC([a,b])\) with \(G(a)=g_a\), \(G(b)=g_b\), and \(\underline f\le G'\le\overline f\) a.e. By the fundamental theorem of calculus, we have \[g_a+\underline f(s-a)\le G(s)\le g_a+\overline f(s-a) \quad \text{and} \quad g_b-\overline f(b-s)\le G(s)\le g_b-\underline f(b-s).\]
    Consequently, \(\underline{G} \leq G \leq \overline{G}\) on \([a,b]\). Finally, we deduce \ref{lemmait3} from \(\overline{G}\) (\(\underline G\)) being the minimum (maximum) of two affine functions.
\end{proof}

\begin{proof}[Proof of \Cref{prop:ass3}]
We prove maximization--minimization follows by replacing \(\varphi\) with \(-\varphi\). If \(m=0\), we must have \(f=\underline f\) a.e.; and when \(m=L\) necessarily \(f=\overline f\). In both cases the claim is immediate, and so we now assume \(0<m<L\).

Recall that \(s_\varphi\in(s_0,s_1)\) is \(\varphi\)'s ``crossing point,'' and define the set of feasible intermediate values
\[
\Gamma\coloneqq\left[A(Q_0)+\underline f(s_\varphi-s_0),\,A(Q_0)+\overline f(s_\varphi-s_0)\right]\cap\left[A(Q_1)-\overline f(s_1-s_\varphi),\,A(Q_1)-\underline f(s_1-s_\varphi)\right].
\]
If \(f\in\mathcal F\), then \(\gamma_f\coloneqq G_f(s_\varphi)\in\Gamma\). In particular, since \(\mathcal F\neq\emptyset\), we have \(\Gamma\neq\emptyset\).

Take arbitrary \(\gamma\in\Gamma\) and define \(G_\gamma\in AC([s_0,s_1])\) by gluing the two-endpoint envelopes from \Cref{lem:grad-envelope} at \(s_\varphi\), with endpoints matched at \(\gamma\).
There are two cases.

\medskip
\noindent
Case 1: \(\varphi>0\) a.e. on \([s_0,s_\varphi)\) and \(\varphi<0\) a.e. on \((s_\varphi,s_1]\).
Define
\[
G_\gamma(s)\coloneqq
\begin{cases}
\min\left\{A(Q_0)+\overline f(s-s_0),\,\gamma-\underline f(s_\varphi-s)\right\}, \quad &\text{if} \quad s\in[s_0,s_\varphi],\\[4pt]
\max\left\{\gamma+\underline f(s-s_\varphi),\,A(Q_1)-\overline f(s_1-s)\right\}, \quad &\text{if} \quad s\in[s_\varphi,s_1].
\end{cases}
\]
Let \(f_\gamma\coloneqq(G_\gamma)'\) a.e. Applying \Cref{lem:grad-envelope} on each subinterval, \(f_\gamma\in\mathcal F\), and for any \(f\in\mathcal F\) with \(G_f(s_\varphi)=\gamma\),
\[
G_f(s)\le G_\gamma(s)\ \text{for all }s\in[s_0,s_\varphi],
\qquad \text{and} \qquad
G_f(s)\ge G_\gamma(s)\ \text{for all }s\in[s_\varphi,s_1].
\]
Since \(A^{-1}\) is increasing, necessarily
\[
\int_{s_0}^{s_\varphi}\varphi(s)A^{-1}\left(G_f(s)\right) ds\le\int_{s_0}^{s_\varphi}\varphi(s)A^{-1}\left(G_\gamma(s)\right) ds,
\]
and, because \(\varphi<0\) a.e. on \((s_\varphi,s_1]\),
\[
\int_{s_\varphi}^{s_1}\varphi(s)A^{-1}\left(G_f(s)\right) ds\le\int_{s_\varphi}^{s_1}\varphi(s)A^{-1}\left(G_\gamma(s)\right) ds.
\]
Therefore, \(I(f)\le I(f_\gamma)\) among all \(f\in\mathcal F\) with \(G_f(s_\varphi)=\gamma\).
Moreover, \(f_\gamma\) EATS.

\medskip
\noindent
Case 2: \(\varphi<0\) a.e. on \([s_0,s_\varphi)\) and \(\varphi>0\) a.e. on \((s_\varphi,s_1]\).
Repeat the construction but swap the envelope roles on each side of \(s_\varphi\).
The same monotonicity argument implies \(I(f)\le I(f_\gamma)\) among all \(f\in\mathcal F\) with \(G_f(s_\varphi)=\gamma\). Likewise, \(f_\gamma\) EATS.

\medskip
\noindent
We conclude that for each \(\gamma\in\Gamma\), the best \(f\) conditional on \(G_f(s_\varphi)=\gamma\) can be chosen to satisfy the \Cref{prop:ass3} structure. Accordingly,
\(\sup_{f\in\mathcal F}I(f)=\sup_{\gamma\in\Gamma}I(f_\gamma)\). 

\begin{claim}
    The function \(\gamma\mapsto I(f_\gamma)\) is Lipschitz continuous on \(\Gamma\). 
\end{claim}
\begin{proof}
    Since \(f_\gamma=G_\gamma'\) a.e. and \(G_\gamma(s_0)=A\left(Q_0\right)\), we have \(G_{f_\gamma}=G_\gamma\). Fix \(\gamma,\tilde\gamma\in\Gamma\). By construction, for each \(s\), \(G_\gamma(s)\) is obtained from affine functions of \(\gamma\), with coefficients in \(\left\{0,1\right\}\), using only \(\min\) and \(\max\). Hence,
\[
\left|G_\gamma(s)-G_{\tilde\gamma}(s)\right|\le\left|\gamma-\tilde\gamma\right| \quad \forall \ s\in\left[s_0,s_1\right].
\]
Moreover, \(G_\gamma'(s)\in\left\{\underline f,\overline f\right\}\) a.e. and \(\overline f\le0\), so \(G_\gamma\) is nonincreasing and
\(G_\gamma(s)\in\left[A\left(Q_1\right),A\left(Q_0\right)\right]\) for all \(s\in\left[s_0,s_1\right]\). Since \(A'\) is continuous and positive on \(\left[Q_1,Q_0\right]\), \(A^{-1}\) is Lipschitz on \(\left[A\left(Q_1\right),A\left(Q_0\right)\right]\). Let
\(L_A\coloneqq\frac{1}{\min_{q\in\left[Q_1,Q_0\right]}A'\left(q\right)}\). Then,
\[
\begin{aligned}
\left|I(f_\gamma)-I(f_{\tilde\gamma})\right|
&\le\int_{s_0}^{s_1}\left|\varphi(s)\right|\left|A^{-1}\left(G_\gamma(s)\right)-A^{-1}\left(G_{\tilde\gamma}(s)\right)\right|ds\\
&\le L_A\left\|\varphi\right\|_{1}\left|\gamma-\tilde\gamma\right|.
\end{aligned}
\]
Since \(\varphi\in L^1\left(\left[s_0,s_1\right]\right)\), this proves that \(\gamma\mapsto I(f_\gamma)\) is Lipschitz continuous on \(\Gamma\).
\end{proof}

The function \(\gamma\mapsto I(f_\gamma)\) is continuous on \(\Gamma\), and \(\Gamma\) is compact, so the supremum is attained at some \(\gamma^*\in\Gamma\).
Let \(f^*\coloneqq f_{\gamma^*}\). Finally, since \(f^*(s)\in\left\{\underline f,\overline f\right\}\) a.e. and \(\int_{s_0}^{s_1}f^*(s) ds=\Delta A\), we have
\(\Delta A=\int_{s_0}^{s_1}f^*(s) ds=\underline f L+\left(\overline f-\underline f\right)\lambda(S_{f^*})\),
so \(\lambda(S_{f^*})=m\).
\end{proof}

\end{document}